\documentclass[runningheads, orivec]{llncs}
\usepackage[T1]{fontenc}
\usepackage{graphicx}
\usepackage{amsmath, amsthm, amssymb, amsfonts, hyperref, xcolor, aligned-overset, bbm}
\usepackage{orcidlink}
\usepackage{cite}
\usepackage{wrapfig}
\newcommand{\bR}{\mathbb R}
\newtheorem{assumption}{Assumption}
\begin{document}
	\title{Tree inference with varifold distances}
	%
	%
	\author{Elodie Maignant \orcidlink{0000-0003-3006-5174} \and
			Tim Conrad \orcidlink{0000-0002-5590-5726} \and
            Christoph von Tycowicz \orcidlink{0000-0002-1447-4069}}
	\authorrunning{E. Maignant et al.}
	%
	\institute{Zuse Institute Berlin, Takustraße 7, 14195 Berlin, Germany \\ \email{\{maignant,conrad,vontycowicz\}@zib.de}}
	\maketitle              
	\begin{abstract}
        In this paper, we consider a tree inference problem motivated by the critical problem in single-cell genomics of reconstructing dynamic cellular processes from sequencing data. In particular, given a population of cells sampled from such a process, we are interested in the problem of ordering the cells according to their progression in the process. This is known as trajectory inference. If the process is differentiation, this amounts to reconstructing the corresponding differentiation tree. One way of doing this in practice is to estimate the shortest-path distance between nodes based on cell similarities observed in sequencing data. Recent sequencing techniques make it possible to measure two types of data: gene expression levels, and RNA velocity, a vector that predicts changes in gene expression. The data then consist of a discrete vector field on a (subset of a) Euclidean space of dimension equal to the number of genes under consideration. By integrating this velocity field, we trace the evolution of gene expression levels in each single cell from some initial stage to its current stage. Eventually, we assume that we have a faithful embedding of the differentiation tree in a Euclidean space, but which we only observe through the curves representing the paths from the root to the nodes. Using varifold distances between such curves, we define a similarity measure between nodes which we prove approximates the shortest-path distance in a tree that is isomorphic to the target tree.
		\keywords{varifold  \and tree \and single-cell.}
	\end{abstract}
	
	\section{Introduction}
	    RNA sequencing technologies quantify RNA molecules in biological samples to determine gene expression levels. Single-cell RNA sequencing (scRNA-seq) enables such information to be extracted at the level of each individual cell. In practice, the gene expression of one cell is displayed as a vector, each coordinate of which contains the expression level of a given gene. Variation in gene expression across a population of cells reflects an underlying dynamic biological process, such as development, or disease progression. It is this process that we are interested in characterizing. However, because single-cell sequencing is a destructive technique, the stage in the process at which each cell was observed is not known. To reconstruct the process, one must order the cells along a developmental or temporal trajectory based on similarities in their expression profiles. This is known in the literature as trajectory inference \cite{saelens_comparison_2019}. Assuming that the underlying biological process is differentiation -- a hierarchical phenomenon -- cell similarities should exhibit a tree structure. The trajectory inference problem is then a tree inference problem. This problem can be decomposed into two sub-problems. First, to construct a similarity measure between cells based on the sequencing data, for example by comparing gene expression levels between cells gene by gene. Then, to find a tree whose shortest-path distance coincides with the similarity measure on the observed nodes. This second sub-problem is standard and well-studied \cite{kalaghatgi_family_2016}. Recent advances in scRNA-seq allow the distinction between spliced and unspliced RNA, enabling estimation of RNA velocity -- a vector that indicates the direction and rate of changes in gene expression levels for each cell \cite{manno_rna_2018,gorin_rna_2022}. Among the approaches that take RNA velocity into account to solve the trajectory inference problem, several rely on similarity measures that assess only local consistency between velocities and gene expression vectors \cite{zhang_inference_2021, lange_cellrank_2022}. Instead, if we interpret the data as a discrete vector field on a subset of a Euclidean space -- of dimension equal to the number of genes under consideration -- and integrate it, we can then recover for each cell a curve describing the evolution of its gene expression levels from some initial stage to its current stage, and compare such global observations. The key idea here is that the relative position of these curves in the space should reflect the hierarchy between the paths from the root to the nodes in the differentiation tree. 
    \begin{figure}[!h]
        \centering
        \vspace{.2cm}
        \includegraphics[width=.72\textwidth]{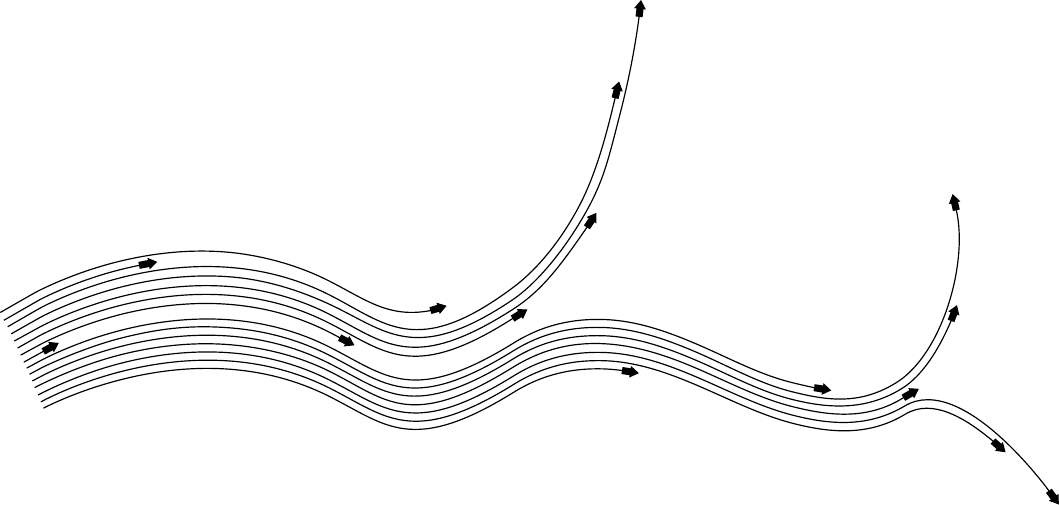}
        \vspace{.4cm}
        \caption{Schematic representation of the integral curves of a RNA velocity field.} 
        \label{fig:velocity}
    \end{figure}
    
        \noindent If so, we show that kernel metrics between such curves represented as oriented varifolds \cite{kaltenmark_general_2017} define a similarity measure between the nodes of the tree that is consistent with its topology. The paper is structured as follows. In the next section, we formalize the tree inference problem we have just identified and we outline our solution. In the third section, we provide some theoretical background on distances between curves represented as oriented varifolds. In the fourth section, we derive from these metrics a similarity measure between nodes, as announced, and we prove it approximates the shortest-path distance in a tree that is isomorphic to the target tree.
	
	\section{Description of the problem}
         Let $T = (V, E, r)$ be a rooted tree embedded into some $\bR^n$ via the two injective maps
    \begin{equation}
            \begin{array}{ccl}
                i \in V \mapsto x_i \in \bR^n
            \end{array} 
            \qquad \text{and} \qquad
            \begin{array}{ccl}
                (i, j) \in E \mapsto [x_i, x_j] \subset \bR^n
            \end{array} 
    \end{equation}
        where $[x_i, x_j]$ denotes here a smooth curve segment of $\bR^n$ connecting $x_i$ and $x_j$ which \textbf{need not be the line segment} from $x_i$ to $x_j$. 
    \begin{figure}[!h]
        \centering
        \includegraphics[width=.7\textwidth]{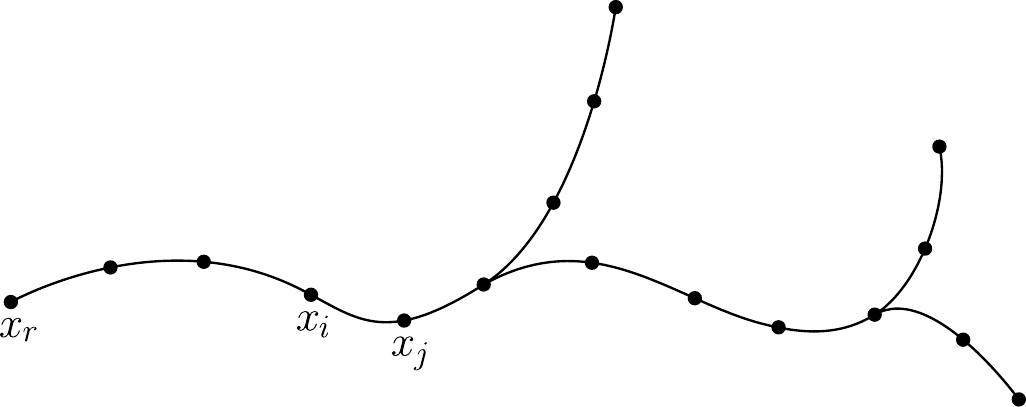}
        \caption{A rooted tree embedded in the plane.} 
        \label{fig:tree}
    \end{figure}
    
        \noindent We ask for the embedding to preserve the topology of the tree, that is for the respective embeddings of two edges to intersect only if the edges share a common node, and only at the image of this common node. Additionally, we require that the embedding of the path from the root $r$ to any node is also a smooth curve, and we extend the notation $[x, y]$ to the curve segment joining any two points $x$ and $y$ on the embedding of a common path from the root to some node. In particular, the map
        \begin{equation}
            i \in V \to [x_r, x_i] \subset \bR^n
        \end{equation}
        should denote the embedding of the path from the root $r$ to the node $i$. Finally, we assume that the embedding of any such path is a regular and orientable curve, so that it defines an \textit{oriented varifold} \cite{kaltenmark_general_2017}. A natural orientation would be, for example, from the root to the end node. We are then interested in reconstructing $T$ from the maps $i \to [x_r, x_i]$ alone.
    \medbreak   
        A straightforward approach to solving this problem would also be to reconstruct the hierarchy between any two nodes $i$ and $j$ by testing whether $x_i$ belongs to the path from the root to $x_j$ or conversely, for example. There are two main reasons for which we do not follow such an approach. First, we might not observe all the nodes of $T$ in practice. Second, the test described above is not very robust to noise, especially in the case of high-dimensional observations. In other words, this approach does not apply to our scenario, pictured in Figure \ref{fig:velocity}.
    \medbreak    
        Instead, we propose to estimate the shortest-path distance in the tree $T$. As we discussed in the introduction, this is a common approach. The shortest-path distance quantifies how far two nodes are from their most recent common ancestor. Equivalently, we propose to quantify how long it takes for two paths to connect in $\bR^n$ relying on distances between curves. Provided that the $n$-dimensional embedding of $T$ reflects its topology well enough, such quantity should recover, if not the shortest-path distance on $T$, a distance that is topologically equivalent (and characterizes $T$ just as much). To ensure this, we make three further assumptions.
    \begin{assumption}
        \label{A1} 
        The distance between two adjacent edges of the same path is increasing as they move further from their common node. Let $i, j, k$ be three nodes of $T$ such that $i < j < k$ meaning $i$ is the parent node of $j$ which is itself the parent node of $k$. Then for all $x, x' \in [x_i, x_j]$ and $y, y'\in [x_j, x_k]$ we have that
    \begin{equation*}
        \ell([x, y]) \leq \ell([x', y']) \Rightarrow \|x - y\| \leq \|x' - y'\|
    \end{equation*}   
        where $\ell(\gamma)$ denotes the arc length of the curve $\gamma$.
    \begin{figure}[!h]
        \centering
        \includegraphics[width=.3\textwidth]{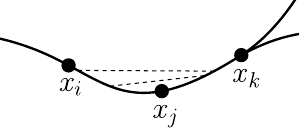}
        \vspace{-.5cm}
        \label{fig:hyp1}
    \end{figure}
    \end{assumption}
    \begin{assumption}
        \label{A2}
        The angle between two adjacent edges of branching paths is increasing as they move further from their common node. Let $i, j, k$ be three nodes of $T$ such that $i < j$ and $i<k$. Then for all $x, x' \in [x_i, x_j]$ and $y, y'\in [x_i, x_k]$ we have that
    \begin{equation*}
        \ell([x_i, x]) +  \ell([x_i, y]) \leq \ell([x_i, x']) + \ell([x_i, y'])\Rightarrow \|\vec{t}(x) - \vec{t}(y)\| \leq \|\vec{t}(x') - \vec{t}(y')\|
    \end{equation*}  
        where the map $x \mapsto \vec{t}(x)$ is defined in the context of oriented varifolds and denotes the unit oriented tangent vector to the curve it is evaluated on.
   \begin{figure}[!h]
        \centering
        \includegraphics[width=.175\textwidth]{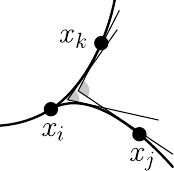}
        \label{fig:hyp2}
        \vspace{-.5cm}
    \end{figure}
    \end{assumption}
    \begin{assumption}
        \label{A3}
        Branching paths deviate sufficiently fast from their initial direction. Let $i, j, k$ be three nodes of $T$ such that $i < j$ and $i<k$. Then there exists some neighborhood $U_i$ of $x_i$ and $0 < a < 2$ such that for all $x \in U_i \cap [x_i, x_j]$ and all $x \in U_i\cap[x_i, x_k]$ we have
    \begin{equation*}
        \|\vec{t}(x) - \vec{t}(x_i)\| \geq \ell([x_i, x])^a.
    \end{equation*}
    \end{assumption}
    The first two assumptions are quite natural and essentially state that we can infer from the embedding whether three consecutive nodes are aligned along a common path originating from the root or not. The last assumption is stronger than the second, primarily technical, and guarantees that we can accurately identify branching nodes on the basis of the tangent vectors alone. Ultimately, just as the approach and results we present in this paper, the three assumptions are only really necessary when the data are noisy and the tree looks more like Figure \ref{fig:velocity} than Figure \ref{fig:tree}.
        
	\section{Distances between curves as oriented varifolds}
        \label{sec:varifold}
        For the choice of a distance between curves, we turn to varifold distances. The reasons for this choice are threefold: (i) The distance between two curves that coincide on part of their domain is simply the distance between the disjoint domains. (ii) Varifold distances are robust to deformations. (iii) Varifold distances are invariant to parameterization while remaining computationally efficient. We follow the framework by Kaltenmark et al.~\cite{kaltenmark_general_2017} known as \textit{oriented varifolds}. This representation provides a foundation for a general class of similarity metrics that apply to curves (and surfaces) as well as to reunions thereof.
    \medbreak  
        Let $X$ be a smooth curve of finite length embedded in $\mathbb{R}^n$ and $W$ be a space of test functions on $\mathbb{R}^n \times \mathbb{S}^{n-1}$, the oriented varifold $\mu_X$ associated to $X$ is the element of the dual $W^\ast$ given by 
    \begin{equation}
        \mu_X: W \to \mathbb{R}, \;\omega \mapsto \int_X \omega(x,\vec{t}(x)) \;d\ell(x),
    \end{equation}
        where the unit length vector $\vec{t}(x) \in \mathbb{S}^{n-1}$ encodes the oriented tangent space to $X$ at $x$. The representation $\mu_X$ depends on the shape and orientation of tangent spaces but not on a choice of a parameterization. 
    \medbreak    
        Furthermore, for distinct curves $X$ and $Y$ we have that $\mu_{X \cup Y} = \mu_X + \mu_Y$. This additivity property gives us in particular (i). Distances on $W^\ast$ and thus between curves are induced from a \textit{Reproducing Kernel Hilbert Space} structure on the set of test functions $W$ defined by a positive definite kernel $k : (\bR^n \times \mathbb{S}^{n-1}) \times (\bR^n \times \mathbb{S}^{n-1}) \to \mathbb{R}$. In particular, the reproducing property guarantees that evaluation functionals $\delta_{(x,t)}: \omega \mapsto \omega(x,t)$ can be represented in terms of the inner product $\delta_{(x,t)}(\omega) = \langle k_{(x,t)}, \omega \rangle_W$ with the reproducing kernel $k_{(x,t)}(\cdot) = k((x,t),\cdot)$. This allows us to define an inner product on the dual by $\langle \delta_{(x_1,t_1)}, \delta_{(x_2,t_2)} \rangle_{W^\ast} = \langle k_{(x_1,t_1)}, k_{(x_2,t_2)}\rangle_W = k((x_1,t_1),(x_2,t_2))$, where the last equality follows from the construction of $W$ from $k$. 
    \medbreak
        While there is a wide range of feasible kernels, we restrict ourselves to tensor products of Gaussian kernels $k_\sigma(x,y)=e^{-\|x - y\|^2/\sigma^2}$ that are of the form $k((x_1,t_1), (x_2,t_2))=k_{\sigma_x}(x_1,x_2)k_{\sigma_t}(t_1,t_2)$. Then, we write the inner product for oriented varifolds as  
	\begin{align}
        \langle \mu_{X}, \mu_{Y}\rangle_{W^\ast} &= \Bigl\langle \int_X \delta_{(x,\vec{t}(x))} d\ell(x), \int_Y \delta_{(y,\vec{t}(y))} d\ell(y)\Bigr\rangle_{W^\ast} \\
        &= \iint_{X \times Y} e^{-\frac{\|x - y\|^2}{\sigma_x^2}} e^{-\frac{\|\vec{t}(x) - \vec{t}(y)\|^2}{\sigma_t^2}} \;d\ell(x)d\ell(y).
	\end{align}
         The separable structure enhances the interpretability of contributions from spatial and orientation characteristics. Intuitively, $k_{\sigma_x}$ measures proximity between point positions, whereas $k_{\sigma_t}$ quantifies the proximity between the associated tangent spaces. For our choice of kernel, the identification $X \to \mu_X$ is injective such that the restriction of the metric on $W^\ast$ defines a distance on the space of curves:
    \begin{equation}
		d_{W^\ast}(X, Y) = \|\mu_{X} - \mu_{Y}\|_{W^\ast}.
	\end{equation}
        This distance can be efficiently approximated due to closed-form expressions for discrete curves such as polygonal chains and is insensitive to the sampling of discrete curves due to the parameterization invariance. Finally, it is robust with respect to small diffeomorphic deformations. Precisely, for two smooth curves $X$ and $Y$ we have
	\begin{equation}
		d^2(X, \phi.Y) \underset{\substack{\vspace{.05cm} \\ \|\phi - id\|_\infty \to 0 \vspace{.05cm} \\ \|d\phi - I\|_F \to 0}}{\to} d^2(X, Y).
	\end{equation}
	    Similar properties have been proven in related literature \cite{glaunes_these_2005, charlier_fshape_2017}. In practice, for the distance between $X$ and $Y$ to be robust to a deformation $\phi$ of $Y$, we need to choose $\sigma_x$ significantly greater than $\|\phi - id\|_\infty$ and $\sigma_t$ significantly greater than $\|d\phi - I\|_F$.
	
	\section{Construction of a shortest-path distance}
	   	Let us finally introduce the following similarity measure between the nodes of the tree $T$
	\begin{equation}
		\Delta(i, j) = d_{W^\ast}([x_r, x_i], [x_r, x_j])^2
	\end{equation}
		where we recall that $[x_r, x_i]$ and $[x_r, x_j]$ denote the embedding in $\bR^n$ of the paths in $T$ from the root $r$ to nodes $i$ and $j$ respectively. Then $\Delta$ approximates the shortest-path distance on a weighted tree $T'$ that is isomorphic to $T$ in the following sense:
	\begin{proposition}
        \label{prop:shortest_path_distance}
		Let $i, j$ be two nodes of $T$. Let $i=i_1,\dots,i_N=j$ denote the shortest path from $i$ to $j$. Then we have the equivalence
	\begin{equation}
		\Delta(i, j) \underset{\substack{\vspace{.05cm} \\ \sigma_x, \sigma_t \to 0 \vspace{.1cm} \\ \sigma_x \asymp \sigma_t}}{\sim} \sum_{k=1}^{N-1}\Delta(i_k, i_{k+1})
	\end{equation}
        where $\sigma_x \asymp \sigma_t$ means here that there exist $k_1, k_2 > 0$ such that $k_1 \sigma_x \leq \sigma_t \leq k_2 \sigma_x$.
	\end{proposition}
		\noindent Note that this last condition is essentially technical and allows us to control the convergence of $\sigma_x$ and $\sigma_t$ - acting at two different scales - simultaneously. The result is a direct consequence of the two following lemmas.
	\begin{lemma}
		Let $i, j$ and $k$ be three adjacent nodes of $T$ such that $i < j < k$. Assume that Assumption \ref{A1} holds. Then we have
	\begin{equation}
        \label{eq:L1}
		\langle \mu_{[x_i, x_j]}, \mu_{[x_j, x_k]}\rangle_{W^\ast} \underset{\substack{\vspace{.05cm} \\ \sigma_x, \sigma_t \to 0 \vspace{.1cm} \\ \sigma_x \asymp \sigma_t}}{=} o\left(\|\mu_{[x_i, x_j]}\|_{W^\ast}^2 + \|\mu_{[x_j, x_k]}\|_{W^\ast}^2\right).
	\end{equation}
	\end{lemma}
	\begin{proof}
        First, notice that all the quantities involved in Equation \ref{eq:L1} are positive. Therefore, we only need to prove that the quotient of the left-hand side scalar product by the right-hand side sum is bounded from above by a function that converges to $0$ when $\sigma_x$ does. 
    \medbreak 
        Let us start with bounding the left-hand side from above. Let $\varepsilon > 0$ and let $x_{j, \varepsilon}^- \in [x_i, x_j]$ and $x_{j, \varepsilon}^+ \in [x_j, x_k]$ such that
    \begin{equation*}
		\ell([x_{j, \varepsilon}^-, x_j])=\ell([x_j, x_{j, \varepsilon}^+]) = \varepsilon.
	\end{equation*} 
        Thanks to the additivity of varifolds, we can decompose the left-hand side scalar product as follows 
	\begin{align*}
		\langle \mu_{[x_i, x_j]}, \mu_{[x_j, x_k]} \rangle_{W^\ast} 
	   	= &\: \langle \mu_{[x_{j, \varepsilon}^-, x_j]} , \mu_{[x_j, x_{j, \varepsilon}^+]} \rangle_{W^\ast} \\
        &+ \langle \mu_{[x_{j, \varepsilon}^-, x_j]} , \mu_{[x_{j, \varepsilon}^+, x_k]} \rangle_{W^\ast} \\
		&+ \langle \mu_{[x_i, x_{j, \varepsilon}^-]}, \mu_{[x_j, x_{j, \varepsilon}^+]} \rangle_{W^\ast} \\
        &+ \langle \mu_{[x_i, x_{j, \varepsilon}^-]}, \mu_{[x_{j, \varepsilon}^+, x_k]} \rangle_{W^\ast}.
	\end{align*}
		Leveraging Assumption \ref{A1}, we derive then the following inequality 
    \begin{align*}
		\langle \mu_{[x_i, x_j]}, \mu_{[x_j, x_k]} \rangle_{W^\ast} 
	   	\leq &\: \varepsilon^2 \\
        &+ \ell([x_j, x_k]) \varepsilon e^{-\|x_j - x_{j, \varepsilon}^+\|^2 / \sigma_x^2} \\
		&+ \ell([x_i, x_j]) \varepsilon e^{-\|x_{j, \varepsilon}^- - x_j\|^2 / \sigma_x^2} \\
		&+ \ell([x_i, x_j]) \ell([x_j, x_k]) e^{-\|x_{j, \varepsilon}^- - x_{j, \varepsilon}^+\|^2 / \sigma_x^2}.
	\end{align*}
        Now, because the paths are assumed to be smooth and regular, we have the following result on the length of an infinitesimal arc
	\begin{equation*}
		\forall x,y \in [x_i, x_k], \: \|x -y\| \underset{x \to y}{\sim} \ell([x, y]).
	\end{equation*}
        The previous upper bound is then of the form
	\begin{align*}
        \langle \mu_{[x_i, x_j]}, \mu_{[x_j, x_k]} \rangle_{W^\ast} 
		&\leq \varepsilon^2 + C_1 \varepsilon e^{-\frac{\varepsilon^2 + o(\varepsilon^2)}{\sigma_x^2}} + C_2e^{-4\frac{\varepsilon^2 + o(\varepsilon^2)}{\sigma_x^2}}
	\end{align*}
		where $C_1$ and $C_2$ do not depend on either $\varepsilon$, $\sigma_x$ or $\sigma_t$. 
    \medbreak  
        Then, let us bound the two norms in the right-hand side from below. For any $0 < \rho < \min(\ell([x_i, x_j]), \ell([x_j, x_k]))$, we can write
	\begin{equation*}
		\|\mu_{[x_i, x_j]}\|_{W^\ast}^2 \geq \iint_{[x_i, x_j]^2}\mathbbm{1}_{\{y\in B(x, \rho)\}}e^{-\frac{\|x - y\|^2}{\sigma_x^2}}e^{-\frac{\|\vec{t}(x) - \vec{t}(y)\|^2}{\sigma_t^2}}d\ell(x)d\ell(y)
    \end{equation*}
        The curves $[x_i, x_j]$ and $[x_j, x_k]$ are assumed to be orientable, that is $x \mapsto \vec{t}(x)$ is a smooth map on both domains. Therefore it is $K$-Lipschitz for some $K>0$ and we have
    \begin{equation*}
		\|\mu_{[x_i, x_j]}\|_{W^\ast}^2 \geq \iint_{[x_i, x_j]^2}\mathbbm{1}_{\{y\in B(x, \rho)\}}e^{-\frac{\rho^2}{\sigma_x^2}}e^{-\frac{K^2\rho^2}{\sigma_t^2}}d\ell(x)d\ell(y)\geq C_3\rho e^{-\frac{\rho^2}{\sigma_x^2}}e^{-K^2\frac{\rho^2}{\sigma_t^2}} 
	\end{equation*}
	    and similarly $\|\mu_{[x_j, x_k]}\|_{W^\ast}^2 \geq C_4\rho e^{-\frac{\rho^2}{\sigma_x^2}}e^{-K^2\frac{\rho^2}{\sigma_t^2}}$ where $C_3$ and $C_4$ do not depend on either $\rho$, $\sigma_x$ or $\sigma_t$.
	\medbreak
        Finally, let us fix $\varepsilon=\sigma_x^p$ for some $\frac{1}{2} < p < 1$ and $\rho=\sigma_x$. Then we obtain an upper bound on the quotient that converges to $0$ when $\sigma_x$ converges to $0$ with the condition that $\sigma_x \asymp \sigma_t$. \qedsymbol
	\end{proof}
	\begin{lemma} 
		Let $i, j$ and $k$ be three adjacent nodes of $T$ such that $i < j$ and $i < k$. Assume that Assumptions \ref{A2} and \ref{A3} hold. Then we have
	\begin{equation}
        \label{L2}
		\langle \mu_{[x_i, x_j]}, \mu_{[x_i, x_k]}\rangle_{W^\ast} \underset{\substack{\vspace{.05cm} \\ \sigma_x, \sigma_t \to 0 \vspace{.1cm} \\ \sigma_x \asymp \sigma_t}}{=} o\big(\|\mu_{[x_i, x_j]}\|_{W^\ast}^2 + \|\mu_{[x_i, x_k]}\|_{W^\ast}^2\big).
	\end{equation}
	\end{lemma}
	\begin{proof}
        Again, notice that all the quantities involved are positive. Therefore, we only need to prove that the quotient of the left-hand side scalar product by the right-hand side sum is bounded from above by a function that converges to $0$ when $\sigma_t$ does. 
    \medbreak    
        Let us bound the left-hand side from above. Let $\varepsilon > 0$ such that $B(x_i, \varepsilon) \subset U_i$ and let $x_{i, \varepsilon}^L \in [x_i, x_j]$ and $x_{i, \varepsilon}^R \in [x_i, x_k]$ such that
    \begin{equation*}
		\ell([x_i, x_{i, \varepsilon}^L])=\ell([x_i, x_{i, \varepsilon}^R]) = \varepsilon.
	\end{equation*} 
        We have that
	\begin{align*}
		\langle \mu_{[x_i, x_j]}, \mu_{[x_j, x_k]} \rangle_{W^\ast} 
	   	= &\: \langle \mu_{[x_i, x_{i, \varepsilon}^L]} , \mu_{[x_i, x_{i, \varepsilon}^R]} \rangle_{W^\ast} \\
        &+ \langle \mu_{[x_i, x_{i, \varepsilon}^L]} , \mu_{[x_{i, \varepsilon}^R, x_k]} \rangle_{W^\ast} \\
		&+ \langle \mu_{[x_{i, \varepsilon}^L, x_j]}, \mu_{[x_i, x_{i, \varepsilon}^R]} \rangle_{W^\ast} \\
        &+ \langle \mu_{[x_{i, \varepsilon}^L, x_j]}, \mu_{[x_{j, \varepsilon}^r, x_k]} \rangle_{W^\ast}.
	\end{align*}
		Leveraging Assumption \ref{A2}, we derive the following inequality
    \begin{align*}
		\langle \mu_{[x_i, x_j]}, \mu_{[x_j, x_k]} \rangle_{W^\ast} 
	   	\leq &\: \varepsilon^2 \\
        &+ \ell([x_i, x_k]) \varepsilon e^{-\|\vec{t}(x_i) - \vec{t}(x_{i, \varepsilon}^R)\|^2 / \sigma_t^2} \\
		&+ \ell([x_i, x_j]) \varepsilon e^{-\|\vec{t}(x_{i, \varepsilon}^L) - \vec{t}(x_i)\|^2 / \sigma_t^2} \\
		&+ \ell([x_i, x_j]) \ell([x_i, x_k]) e^{-\|\vec{t}(x_{i, \varepsilon}^L) - \vec{t}(x_{i, \varepsilon}^R)\|^2 / \sigma_t^2}.
	\end{align*}
        Leveraging then Assumption \ref{A3}, we obtain a lower bound of the form
	\begin{align*}
        \langle \mu_{[x_i, x_j]}, \mu_{[x_j, x_k]} \rangle_{W^\ast} 
		&\leq \varepsilon^2 + C_1\varepsilon e^{-\frac{\varepsilon^{2a}}{\sigma_t^2}} + C_2e^{-4\frac{\varepsilon^{2a}}{\sigma_t^2}}
	\end{align*}
		where $C_1$ and $C_2$ do not depend on either $\varepsilon$, $\sigma_x$ or $\sigma_t$.
        \medbreak
        Then, let us bound the two norms in the right-hand side from below. The map $x \to \vec{t}(x)$ is $K$-Lipschitz on $[x_i, x_j]$ and $[x_i, x_k]$ for some $K>0$, and for any $0 < \rho < \min(\ell([x_i, x_j]), \ell([x_i, x_k]))$, we then have
	\begin{equation*}
		\|\mu_{[x_i, x_j]}\|_{W^\ast}^2 \geq \iint_{[x_i, x_j]^2}\mathbbm{1}_{\{y\in B(x, \rho)\}}e^{-\frac{\rho^2}{\sigma_x^2}}e^{-\frac{K^2\rho^2}{\sigma_t^2}}d\ell(x)d\ell(y)\geq C_3\rho e^{-\frac{\rho^2}{\sigma_x^2}}e^{-K^2\frac{\rho^2}{\sigma_t^2}}
	\end{equation*}
	    and similarly $\|\mu_{[x_i, x_k]}\|_{W^\ast}^2 \geq C_4re^{-\frac{\rho^2}{\sigma_x^2}}e^{-K^2\frac{\rho^2}{\sigma_t^2}}$ where $C_3$ and $C_4$ do not depend on either $\rho$, $\sigma_x$ or $\sigma_t$.
		\medbreak
        Finally, let us fix $\varepsilon=\sigma_x^p$ for some $\frac{1}{2} < p < \frac{1}{a}$ and $\rho=\sigma_x$. Then we obtain an upper bound on the quotient that converges to $0$ when $\sigma_x$ converges to $0$ with the condition that $\sigma_x \asymp \sigma_t$. \qedsymbol
	\end{proof}
	    Let us sketch out the proof of Proposition \ref{prop:shortest_path_distance}. First, we write the varifold distance between the paths (in $\bR^n$) from the root to $i$ and $j$ as the distance between the two paths from their common ancestor, which is one of the nodes visited by the shortest path in $T$. Then we write each of the two paths as the union of the edges composing them and decompose their squared distance accordingly. For all scalar products thus generated, we apply Lemma 1, unless one of the nodes considered is the common ancestor and a branching point, in which case we apply Lemma 2. 
    \medbreak
        As a consequence of Proposition \ref{prop:shortest_path_distance}, we also have that $\Delta$ asymptotically satisfies the triangular inequality
	\begin{equation}
		\Delta(i, k) \underset{\substack{\vspace{.05cm} \\ \sigma_x \to 0 \vspace{.05cm} \\ \sigma_t \to 0}}{\leq} \Delta(i, j) + \Delta(j, k) + o(\sup_{e \in E}\Delta(e))
	\end{equation}
        and the four-points condition for a distance matrix to be a shortest-path distance on a tree \cite{pereira_note_1969}
	  \begin{equation}
	    \Delta(i, k) + \Delta(j, l) \leq \max(\Delta(i, j) + \Delta(k, l), \Delta(j, k) +  \Delta(l, i)) + o(\sup_{e \in E}\Delta(e)).
	\end{equation}
        In practice, the choice of $\sigma_x$ and $\sigma_t$ is a trade-off between precision and robustness. Small enough so that the similarity measure $\Delta$ characterizes the topology of $T$ sufficiently for us to then apply distance-based inference methods of the literature, that is so that the equivalence in Proposition \ref{prop:shortest_path_distance} is close to an equality and the two previous inequalities are almost satisfied. But large enough that $\Delta$ is not too sensitive to noise in the data, as discussed in Section \ref{sec:varifold}.
	
	\section{Conclusion}
	    In this work, we tackled the problem of inferring the topology of a rooted tree embedded in $\bR^n$ from the set of curves defined by the embedding of the paths from the root to the nodes. We showed that kernel metrics between curves represented as oriented varifolds define a similarity measure between the nodes of the tree that is consistent with the topology of the latter. As a next step, we will build upon our new framework to address the reconstruction of cell lineage trees described by the integral curves of a RNA velocity field.
	
	\begin{credits}
		\subsubsection{\ackname} This research is funded by the Deutsche Forschungsgemeinschaft (DFG, German Research Foundation) under Germany´s Excellence Strategy – The Berlin Mathematics Research Center MATH+ (EXC-2046/1, project ID: 390685689) and part of this work was carried out at the Erwin Schrödinger International Institute for Mathematics and Physics (ESI) during the Thematic Programme on Infinite-dimensional Geometry: Theory and Applications.
	\end{credits}
	%
	%
	%
	\bibliographystyle{splncs04}
	\bibliography{GSI25-136.bib}
	
\end{document}